\documentclass[conference,compsoc]{IEEEtran}
\usepackage{amsmath,amssymb,amsthm,booktabs,graphicx,array}
\usepackage{hyperref}
\usepackage{orcidlink}
\usepackage{xurl}
\newtheorem{theorem}{Theorem}

\newtheorem{definition}[theorem]{Definition}
\newcommand{\R}{\mathbb R}
\newcommand{\kernel}{\operatorname{ker}}
\newcommand{\range}{\operatorname{range}}
\newcommand{\proj}{\operatorname{Proj}}
\graphicspath{{figures/}{placeholder/}}

\title{Quantifying and Prioritizing Stealthy Source-Manipulation Threats
Against Multi-Source Anomaly Detectors}

\author{\IEEEauthorblockN{Indraveni Chebolu\orcidlink{0009-0009-6128-6253}}
\IEEEauthorblockA{\textit{Information Security Group} \\
\textit{Centre for Development of Advanced Computing}\\
Hyderabad, Telangana, India \\
indravenik@cdac.in}
\and
\IEEEauthorblockN{Seema Panghal\orcidlink{0009-0008-8943-3537}}
\IEEEauthorblockA{\textit{Information Security Group} \\
\textit{Centre for Development of Advanced Computing}\\
Hyderabad, Telangana, India \\
seemap@cdac.in}
\and
\IEEEauthorblockN{Arnab Mallick\orcidlink{0009-0008-5458-6429}}
\IEEEauthorblockA{\textit{Information Security Group} \\
\textit{Centre for Development of Advanced Computing}\\
Hyderabad, Telangana, India \\
arnabm@cdac.in}
\and
\IEEEauthorblockN{Ch A S Murty\orcidlink{0000-0002-9507-9481}}
\IEEEauthorblockA{\textit{Information Security Group} \\
\textit{Centre for Development of Advanced Computing}\\
Hyderabad, Telangana, India \\
chasmurty@cdac.in}
}

\begin{document}
\maketitle

\begin{abstract}
Security pipelines often fuse several data sources and protect a downstream
model with a cross-source anomaly detector.  Such detectors are usually judged
by their clean false-positive rate and by how well they catch random
corruption.  Neither metric captures an adversary who controls some sources
and targets a specific downstream decision.  We propose a threat-quantification
score, the \emph{task-conditioned visibility margin}, defined as the least
detector evidence a manipulation must reveal per unit of harm to the
downstream task.  The margin has a closed form, takes about $85\,\mu$s per
certificate, yields a minimum-evidence attack, and identifies task-relevant
directions the detector cannot see.  On a 17-channel industrial telemetry benchmark, a
detector that flags 82.8\% of random single-channel corruptions is first-order
blind on 64.7\% of task/subset pairs in an audit of one- to four-channel
subsets.  Task-conditioned manipulations cause more than twice as many falsely
reassuring errors as generic stealthy ones, and more of them evade the alarm on
four of five splits, mainly at small budgets.  Used to prioritize audits on
held-out subsets, the margin finds a median of four verified failures in ten
checks, compared with none for random order and two for alarm-based ranking,
and performs comparably to the task-gradient baseline.  It also ranks vulnerable
configurations of the learned GDN and TranAD detectors and of a five-source
crop-yield panel.
\end{abstract}

\section{Introduction}
Analysts rarely have the budget to verify every data source that feeds a
model.  When several sources are fused and checked by a cross-source anomaly
detector, a practical question follows: \emph{which source subsets should be
audited first?}  Standard detector metrics do not answer it.  A detector can
flag most random corruptions and still miss the coordinated change that a
downstream decision depends on.

The mechanism is geometric.  A detector constrains only the directions its
residual responds to; when the task depends on the remaining directions, a
coordinated change across sources moves the decision while the residual
barely moves.

This is the static, task-conditioned form of unobservable false-data
injection~\cite{liu2009fdi}.  We turn it into a per-subset threat score for
learned, nonlinear detectors and downstream models, and use the score to rank
audits.

\paragraph{Contributions}
\begin{itemize}
\item \textbf{Threat score.} For any compromised source subset, the margin
$\Gamma_M$ measures the minimum whitened detector evidence per unit of
downstream harm.  It has a closed form, gives a minimum-evidence attack, and
equals the slope of the optimal nonlinear evidence--harm trade-off for small
attacks.
\item \textbf{Explanation.} The margin separates exact blind spots from weakly
visible directions and shows why the same detector carries different risk for
different downstream tasks.
\item \textbf{Prioritization.} On held-out source subsets, ranking by
$\Gamma_M$ finds verified failures ahead of random and alarm-based orderings
and comparably to the task-gradient baseline, while also diagnosing which
detector directions make the subsets vulnerable.
\item \textbf{Evaluation.} We test principal-component-analysis (PCA),
reconstruction, GDN, and TranAD detectors on industrial telemetry and a crop-yield panel with adaptive
finite attacks.  The artifact~\cite{arnab_mallick_2026_23682852} is public.
\end{itemize}

\section{Threat Model}
The state $x\in\R^k$ concatenates source blocks.  The detector computes a
residual vector $g(x)\in\R^r$ and raises an alarm when a score of $g(x)$
crosses a threshold calibrated at the 95th validation percentile.  The
downstream task is a differentiable scalar $f(x)$, such as a regression output
or class margin, signed so that increasing it is harmful.

The adversary controls a declared subset $M$ of sources, for example through a
compromised acquisition, calibration, or reporting path.  Its actions are
parameterized by $b\in\R^m$ through a map $\phi_{x,M}$ with
$\phi_{x,M}(0)=x$; sources outside $M$ stay fixed.  We write
\[
G_M(b)=g(\phi_{x,M}(b)),\qquad F_M(b)=f(\phi_{x,M}(b)).
\]
The adversary seeks a large task change $F_M(b)-F_M(0)$ with a small detector
change $G_M(b)-G_M(0)$ under a Euclidean norm budget $\|b\|_2\le\rho$.
The defender's objective is to prioritize source subsets for verification.

\paragraph{Controls and attacker knowledge}
Controlling a source means altering the input values delivered during the
evaluated window, not the ground truth, other sources, model parameters, or
alarm threshold.  In Hydraulic telemetry, each compromised channel permits a
gain and a bias on its original trace; summary features are recomputed from
that trace (Appendix~\ref{app:controls}).  This models a compromised input
path, not arbitrary editing of derived features.  The defender computes the
margin using its own models and tests attacks under white-box access;
Section~\ref{sec:graybox} examines transfer to surrogate-model attackers.
Physical compromise and fully black-box attacks are outside these experiments.

\section{Quantifying the Threat}
Let $J_M=DG_M(0)$ be the restricted detector Jacobian, $u_M=\nabla F_M(0)$ the
restricted task gradient, $\Omega\succ0$ the residual weighting (the noise
covariance, or $I$ for standardized residuals), and
$A_M=J_M^\top\Omega^{-1}J_M$ the local detector information.
Here $\proj$ denotes Euclidean orthogonal projection and $^\dagger$
denotes the Moore--Penrose pseudoinverse.

\paragraph{Exact blind spots}
If $\proj_{\kernel J_M}u_M\ne0$, the adversary can change the task while the
detector output stays constant to first order.

\begin{theorem}[Harmful blind manifold]\label{thm:manifold}
If $G_M$ is $C^1$ with constant rank $r_M$ near zero, its level set through
zero is locally a manifold of dimension $m-r_M$ with tangent space
$\kernel J_M$.  If $h_M=\|\proj_{\kernel J_M}u_M\|_2>0$, a unit-speed curve on
this level set changes the task at first-order rate $h_M$.
\end{theorem}

\paragraph{Graded visibility}
Exact blind spots are the limiting case of a broader phenomenon: directions
that the detector observes only weakly.  The margin grades this continuum,
from exactly blind ($\Gamma_M=0$) to strongly observed.

\begin{definition}[Task-conditioned visibility margin]
For $u_M\ne0$,
\begin{equation}
\Gamma_M=\inf_{u_M^\top b=1}\|J_Mb\|_{\Omega^{-1}},\qquad
\|z\|_{\Omega^{-1}}^2=z^\top\Omega^{-1}z.
\label{eq:gamma}
\end{equation}
\end{definition}

A small $\Gamma_M$ means that a unit of first-order task change requires
little linearized detector evidence.  If $u_M=0$, no first-order task change
is possible on that subset, and the case is excluded from ranking.

\begin{theorem}[Closed form and minimum-evidence attack]\label{thm:gamma}
$\Gamma_M=0$ if and only if $h_M>0$.  Otherwise
\begin{equation}
\Gamma_M^2=\bigl(u_M^\top A_M^\dagger u_M\bigr)^{-1},\qquad
b^*(\Delta)=\Delta\,\frac{A_M^\dagger u_M}{u_M^\top A_M^\dagger u_M},
\label{eq:closed}
\end{equation}
and $b^*(\Delta)$ minimizes first-order detector evidence among attacks with
first-order task shift $\Delta$.
\end{theorem}

The margin is a local quantity, and it also governs the true nonlinear
trade-off for small attacks.

\begin{theorem}[Local frontier]
\label{thm:frontier}
Suppose $F_M$ and $G_M$ are $C^2$ and
$\Gamma_M>0$. Fix any constant
$C>\|b^*(1)\|_2$.

For $\Delta>0$, define the minimum
observable evidence as
\begin{equation}
\label{eq:local-frontier}
V_M(\Delta)
:=
\min_{\substack{
F_M(b)-F_M(0)=\Delta\\
\|b\|_2\le C\Delta
}}
\|G_M(b)-G_M(0)\|_{\Omega^{-1}}.
\end{equation}

Then, as $\Delta\downarrow 0$,
\begin{equation}
\boxed{
V_M(\Delta)
=
\Gamma_M\Delta+O(\Delta^2).
}
\end{equation}
\end{theorem}

Proofs are in Appendix~\ref{app:proofs}.  Theorem~\ref{thm:frontier} concerns
small attacks: curvature affects its $O(\Delta^2)$ remainder, while
$\Gamma_M$ alone does not account for a finite norm budget or the available
alarm slack.  A low margin therefore need not imply a feasible undetected
attack; every reported failure is tested against the full nonlinear pipeline.

The margin is invariant under smooth, locally invertible changes of attack
coordinates.  It also cannot decrease when an independent check contributes
a positive-semidefinite information matrix $B_S\succeq0$ to $A_M$.  We
standardize task outputs to compare margins across tasks.

\paragraph{Why task conditioning matters}
The smallest singular value $\sigma_{\min}$ of $J_M$ is task-agnostic,
whereas $\Gamma_M$ uses the task gradient.  Two tasks can therefore share the
same detector and Jacobian spectrum but have different margins
(Table~\ref{tab:task}).

\paragraph{What the task-gradient baseline omits}
The norm $\|u_M\|_2$ measures task change per unit control magnitude but
ignores detector sensitivity.  For example, let $\Omega=I$,
$J_M=\operatorname{diag}(1,\epsilon)$, $0<\epsilon<1$, and let
$e_1,e_2$ be standard basis vectors.  Tasks with gradients $e_1$ and $e_2$
have equal norms but margins $1$ and $\epsilon$: the second is less visible
to the detector.  Conversely, if $A_M=\alpha I$ for a common $\alpha>0$,
then $\Gamma_M=\sqrt{\alpha}/\|u_M\|_2$; sorting by increasing margin is
equivalent to sorting by decreasing task-gradient norm.  The margin's
additional value is to expose \emph{detector visibility}, even when both
scores yield similar empirical audit rankings.

\paragraph{Audit prioritization}
For a fixed task and accepted state, the defender ranks candidate subsets
by increasing $\Gamma_M$ and verifies selected cases with finite attacks.
Each certificate uses a singular-value decomposition (SVD) and a
pseudoinverse.  It costs about $85\,\mu$s, around 70 times less than one
numerical attack optimization (5.9\,ms).

\section{Evaluation Setup}
\paragraph{Industrial telemetry}
The UCI Hydraulic Systems rig~\cite{helwig2015hydraulic} provides 2,205
cycles from 17 channels.  Each channel contributes its mean, standard
deviation, root-mean-square (RMS), minimum, and maximum (85 features).  Four downstream tasks
predict cooler, valve, pump, and accumulator condition.  We use five seeded
contiguous-block 60/20/20 splits with standardized multi-output ridge and
multilayer-perceptron (MLP) task models (balanced accuracy .791 and .914).  The PCA detector projects standardized
features onto the $r$ lowest-variance principal directions, with residual
dimension $r\in\{1,2,4,8,85\}$ ($r=85$ is the full residual), and alarms on
the squared residual norm, the square of the evidence in Eq.~\eqref{eq:gamma}
when $\Omega=I$.  The adversary applies a fractional gain
and a bias to each compromised channel, and all five summaries are recomputed
exactly (Appendix~\ref{app:controls}).

\paragraph{Learned detectors}
We repeat the audit with GDN~\cite{deng2021gdn} and
TranAD~\cite{tuli2022tranad} on per-second 17-channel sequences, with the
native point score aggregated by its maximum over each cycle.  Here $J_M$ is
the central-difference Jacobian (step $10^{-4}$) of the full signed residual
sequence, scaled per channel by its training median and interquartile range
and evaluated at the training-mean cycle.  The max-aggregated alarm is
non-smooth, so for these detectors we use the margin as a ranking score and
validate it with finite attacks.

\paragraph{Crop-yield panel}
The ICRISAT panel~\cite{mohapatra2023icrisatdata} has five sources (rainfall,
temperature, irrigated area, nitrogen fertilizer, acreage) and reported yield
as the target.  We use a temporal split (train through 2008, validate on
2009--2010, test on 2,002 later rows), a ridge or contextual MLP task model,
and a PCA or reconstruction detector.

\paragraph{Attacks}
All attacks share a source support $M$, norm budget $\|b\|_2\le\rho$,
and detector threshold.  A0 uses a random direction; A1 follows the task
gradient; A2 follows the detector's least visible direction regardless of
harm (generic stealth); A4 follows the task-conditioned direction $b^*$ of
Eq.~\eqref{eq:closed}, rescaled to the budget.  The matched comparison uses
the 17 single-channel and 43 evaluated two-channel subsets.  These supports
are first-order identified ($\ker J_M=\{0\}$), so no nonzero control is
perfectly invisible to first order; the artifact~\cite{arnab_mallick_2026_23682852} also implements an
exact-kernel attack (A3) for blind subsets.  Sequential least-squares
programming (SLSQP) supplies a constrained finite-attack reference.
Numerical rank counts singular values above
$10^{-10}\max(m,r)\max(\sigma_1,1)$, where $\sigma_1$ is the largest
singular value of $J_M$.

\paragraph{Outcome}
A \emph{verified failure} is a feasible finite attack that turns an initially
correct, accepted prediction into a falsely reassuring one, i.e., a
healthier condition than the true one (such as full cooler efficiency when
the cooler is degraded), without raising an alarm.  In Table~\ref{tab:matched},
\emph{Adverse} is the rate of newly caused falsely reassuring errors,
\emph{Alarm} is the rate of flagged attacks, and \emph{Joint} is the rate of
adverse errors without alarms, all among initially correct, accepted cases.
Thresholds, supports, and budgets are fixed on training and validation data;
test data is used once.

\section{Results}
\subsection{Detectors That Catch Random Corruption Have Blind Spots}\label{sec:blind}
With residual dimension $r=8$, the Hydraulic detector has a test
false-positive rate (FPR) of 3.18\% and flags 82.8\% of random single-channel gain/bias corruptions at
radius .5.  Over the reference-split audit of source subsets of size one to
four, the same detector is first-order blind on 64.7\% of task/subset pairs.
With $r\le4$ the blind fraction is at least 98.9\%,
as expected: most audited subsets have three or four channels and hence six or
eight controls, more than the $r$ residual coordinates.  The full residual
($r=85$) has no blind pair.

Table~\ref{tab:task} compares the same two-channel subset, TS4 (temperature)
and VS1 (vibration), across four tasks.  The rank and $\sigma_{\min}$ of
$J_M$ are identical, while $\Gamma_M$ spans almost three orders of magnitude;
the margin's ordering agrees with the reported realized A4 harm.  The last
column records the finite A4 task-output change in standardized task units.

\begin{table}[t]\centering\small
\caption{One subset (TS4+VS1), four tasks.  Rank and the smallest
Jacobian singular value are task-agnostic; A4 harm is a realized
standardized task change.}
\label{tab:task}
\setlength{\tabcolsep}{4pt}
\begin{tabular}{lrrrr}\toprule
Task&Rank&$\sigma_{\min}$&$\Gamma_M$&A4 harm\\\midrule
Cooler&4&$7.69\!\times\!10^{-8}$&$1.69\!\times\!10^{-4}$&.0008\\
Valve&4&$7.69\!\times\!10^{-8}$&$5.73\!\times\!10^{-7}$&.0794\\
Pump&4&$7.69\!\times\!10^{-8}$&$2.31\!\times\!10^{-6}$&.0185\\
Accumulator&4&$7.69\!\times\!10^{-8}$&$3.71\!\times\!10^{-7}$&.1301\\\bottomrule
\end{tabular}
\end{table}

\subsection{Task-Conditioned Attacks}
We compare attacks on identical tuples of split, state, subset, task, budget,
and threshold.  The population contains 1,229--1,369 initially correct,
accepted task--cycle units per split, over 17 single-channel and 43
evaluated two-channel subsets and five budgets.

\begin{table}[t]\centering\scriptsize
\setlength{\tabcolsep}{3pt}
\caption{Matched Hydraulic outcomes (\%): five-split median [range].
\emph{Joint} is a newly caused falsely reassuring error without an alarm.}
\label{tab:matched}
\begin{tabular}{lrrr}\toprule
Attack&Adverse&Alarm&Joint\\\midrule
A0 random&18.3 [14.1, 20.5]&95.9 [89.3, 96.5]&0.20 [0.05, 0.74]\\
A1 task gradient&46.3 [40.9, 52.8]&97.9 [94.8, 98.1]&0.22 [0.11, 0.99]\\
A2 generic stealth&6.6 [5.1, 7.3]&36.5 [18.7, 41.7]&2.02 [1.38, 3.25]\\
A4 task conditioned&14.5 [11.9, 16.9]&49.2 [30.9, 53.2]&2.31 [1.76, 5.35]\\\bottomrule
\end{tabular}
\end{table}

Table~\ref{tab:matched} summarizes the outcomes.  A4 causes more adverse
outcomes than A2 on every split (11.9--16.9\% versus 5.1--7.3\%).  A1 has the highest adverse-outcome rate, but its attacks also trigger
detector alarms in at least 94.8\% of evaluated cases.  On
the joint endpoint, A4 exceeds A2 on four of five splits, with the gain
concentrated at budgets .1--.5; at larger budgets both attacks trigger more
alarms and converge (Fig.~\ref{fig:budget}).

The margin also explains where the A4 advantage concentrates.  Low-margin
directions trade detector evidence for control magnitude (median
$\|b^*(1)\|_2=11.4$ over matched cases), so at small budgets A4 achieves its
task effect while staying closest to the residual floor, which matches the
gains at budgets .1--.5 in Fig.~\ref{fig:budget}.  Certificates use one
Jacobian at the training-mean cycle, while alarms are evaluated on each test
cycle; these per-cycle deviations account for alarms beyond the first-order
prediction.

\begin{figure}[t]\centering
\includegraphics[width=\columnwidth]{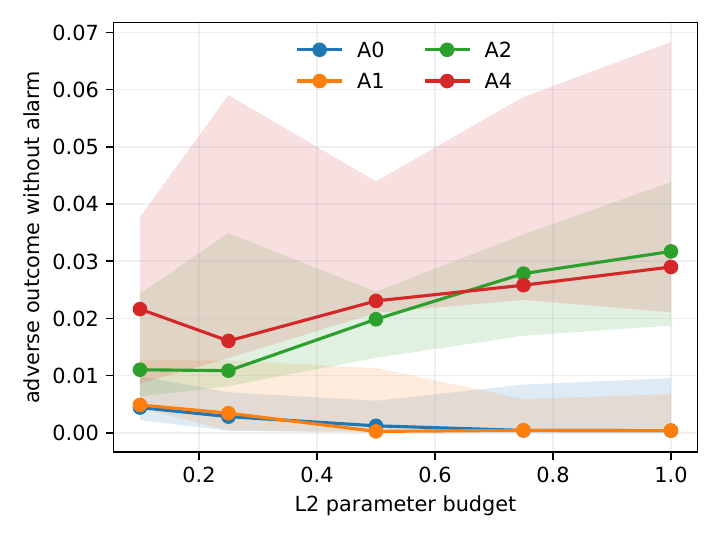}
\caption{Joint adverse-without-alarm rate by attack budget.  Lines are
five-split medians; bands span the splits.}
\label{fig:budget}
\end{figure}

\subsection{Audit Prioritization}
An audit \emph{case} is a tuple of split, source subset, task, and one fixed
initially correct, accepted state.  We partition cases by hashed subset
identity: 520 development cases select tie rules and a combined ranker,
while 680 cases from disjoint subsets are held out.
A finite search at budgets .1, .25, and .5 labels 94 of the held-out cases as
verified failures.  It starts from zero, the task-gradient direction, and
$b^*$, so the labels do not favor either informed ranker.  Alarm slack ranks
cases by their clean distance to the threshold.

\begin{table}[t]\centering\small
\caption{Median held-out verified failures found within an audit budget (94
failures among 680 held-out cases).}
\label{tab:audit}
\begin{tabular}{lrrrr}\toprule
Ranker&10&25&50&100\\\midrule
Random order&0&3&6&13\\
Alarm slack / scalar score&2&3&4&13\\
Task-gradient norm&4&8&13&16\\
Combined (validation-selected)&2&9&13&16\\
$\Gamma_M$ (ours)&4&9&13&15\\\bottomrule
\end{tabular}
\end{table}

\begin{figure*}[t]\centering
\includegraphics[width=.78\textwidth]{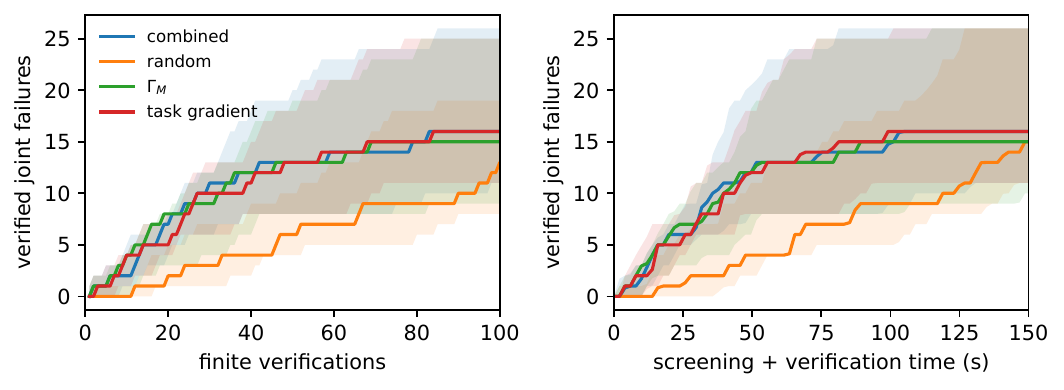}
\caption{Held-out audit allocation by number of verifications (left) and by
screening-plus-verification time (right).  Curves are five-split medians;
bands span the splits.}
\label{fig:audit}
\end{figure*}

Table~\ref{tab:audit} and Fig.~\ref{fig:audit} report the results.  Within
ten audits, $\Gamma_M$ finds a median of four failures, compared with zero for
random order and two for both alarm-slack ranking and the combined ranker.
Its advantage over alarm-slack ranking is 2.2 failures (paired bootstrap
interval $[1.0,3.2]$).  Compared with task-gradient ranking, the reported
paired difference is 0.6 (interval $[0,1.8]$).  This comparison uses
within-split count differences: it need not equal the difference of the
separately reported medians, both of which are four at ten checks.
At budgets 25, 50, and 100, the medians are $9$, $13$, and $15$ for the
margin, versus $8$, $13$, and $16$ for gradient norm
(Table~\ref{tab:audit}).  The two rankers therefore perform comparably on
held-out cases; the margin additionally identifies weakly monitored task
directions (Table~\ref{tab:task}).  Residual whitening by identity,
diagonal, or Ledoit--Wolf covariance gives the same margin-ranker medians.

\subsection{Learned Detectors}
Across four independently trained splits and 128 subset/task combinations,
neither GDN nor TranAD has an exact first-order blind direction, yet margins on
the reference split span from .0029 to 1138 for GDN and from .010 to 3112 for
TranAD.  At budget .5, the Spearman association between the margin and the
joint outcome is negative on all four splits: $-.675$ to $-.229$ for GDN and
$-.655$ to $-.325$ for TranAD.  Task-conditioned attacks expose 24--34 GDN and
24--39 TranAD vulnerable configurations per split; on the reference split,
generic stealth exposes 12 and 19.  Fig.~\ref{fig:external} shows that
auditing in margin order finds these failures faster than random order.

\begin{figure}[t]\centering
\includegraphics[width=\columnwidth]{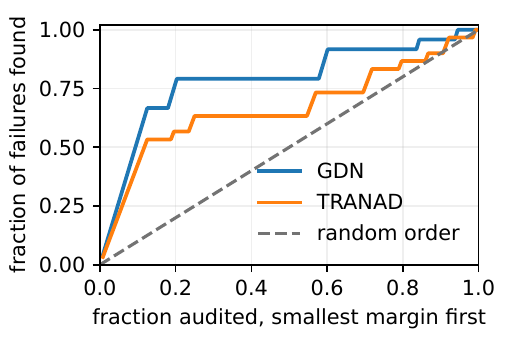}
\caption{GDN and TranAD audit allocation at budget .5, with configurations
ordered by increasing margin.}
\label{fig:external}
\end{figure}

\subsection{Gray-Box Transfer}\label{sec:graybox}
We repeat the Hydraulic attacks with surrogate task models and detectors on 60
fixed two-channel subsets.  White-box attacks have median standardized harm
.155.  Surrogates with the same architecture, trained on 10--100\% of the
data or with independent weights, retain median harm .058 to .156, with
median first-order detector leakage below $10^{-6}$.  Here leakage denotes
the target detector's linearized residual change under the transferred
control, $\|J_Mb\|_{\Omega^{-1}}$.
A different-architecture surrogate
retains harm .089 with leakage $1.7\times10^{-5}$.  Detector mismatch
affects leakage, whereas task-gradient mismatch affects transferred harm.
These tests evaluate imperfect model knowledge for fixed source supports.

\subsection{Crop-Yield Panel and Local Accuracy}\label{sec:crop}
The linear crop detector has six harmful blind subsets among all 31.  Adding
a direct check on fertilizer leaves two, and adding acreage removes all of
them.  With the nonlinear task model and reconstruction detector, the margin
is negatively associated with A4 harm for every one of five model seeds
(Spearman $-.561$ to $-.212$), and A4 exceeds A2 in median harm for every
seed.  The finite-attack optimization also agrees with the local prediction of
Theorem~\ref{thm:frontier}: at task effect $\Delta=.01$, optimized nonlinear
evidence per unit harm differs from $\Gamma_M$ by a median of 5.1\% across
40 fixed states (Spearman .994).

\section{Related Work}

\textbf{Stealthy false-data injection.}
Liu et al.~\cite{liu2009fdi} showed that an attacker who knows the
measurement model of DC state estimation can inject a change in the
column space of the measurement matrix that leaves the bad-data residual
unchanged. Later work quantified and defended against this attack.
Security indices give the fewest measurements an attacker must
corrupt~\cite{sandberg2010security, hendrickx2014security}. Kosut et
al.~\cite{kosut2011malicious} characterize unobservable attacks and design
detectors for them. Bobba et al.~\cite{bobba2010detecting} and Kim and
Poor~\cite{kim2011strategic} choose which measurements to protect so that
no unobservable attack remains. In dynamic systems, Pasqualetti et
al.~\cite{pasqualetti2013attack} characterize undetectable attacks through
system zeros. Fawzi et al.~\cite{fawzi2014secure} bound how many sensors
can be corrupted while the state remains recoverable. Mo and
Sinopoli~\cite{mo2016performance} bound the degradation a stealthy
attacker can cause. The output-to-output gain of Teixeira et
al.~\cite{teixeira2015gain} is closest to our score: for a static scalar
task, $1/\Gamma_M$ plays the same role. Theorem~\ref{thm:manifold} is a nonlinear,
task-restricted analogue of the column-space condition. Our monotonicity
property parallels the measurement-protection results: in our crop panel,
a direct check on acreage removes every harmful blind subset
(Section~\ref{sec:crop}). These works assume a known linear or linearized physical
model and measure harm on a state estimate. Our detectors and task models
are learned from data, harm is measured on a downstream prediction, and
the score is used to rank source subsets for audit.

\textbf{Evading anomaly detectors in industrial systems.}
Urbina et al.~\cite{urbina2016limiting} showed that stealthy attackers can
cause substantial damage against widely deployed industrial control system
(ICS) detectors, and proposed an impact metric under a false-alarm budget.
Erba et al.~\cite{erba2020concealment,erba2023whitebox} optimize
manipulations that conceal real process anomalies from
reconstruction-based detectors, including in white-box settings. Giraldo
et al.~\cite{giraldo2018survey} survey physics-based detection and how it
is evaluated. These works ask whether an existing anomaly can be hidden,
and they evaluate one attack at a time. We ask which source subsets let an
attacker \emph{create} a harmful downstream change, and we give a
closed-form screen per subset before any attack is run.

\textbf{Sensor fusion.}
Attacks on autonomous-driving perception show that fused systems fail when
an adversary controls a subset of sources. Shen et
al.~\cite{shen2020drift} exploit GPS spoofing against multi-sensor
localization, and Cao et al.~\cite{cao2021invisible} construct objects that
are invisible to both camera and LiDAR. Each of these attacks is built by
hand for one system. $\Gamma_M$ is a generic screen for which subsets of a
fused pipeline are weakly cross-checked.

\textbf{Adaptive evaluation and robustness scores.}
Test-time evasion~\cite{biggio2013evasion} and adaptive attacks on
defenses~\cite{carlini2017detectors,athalye2018obfuscated,tramer2020adaptive}
show that detectors must be attacked as complete systems. For this reason,
every failure we report is verified by a finite attack against the full
pipeline. Certified defenses~\cite{wong2018provable,cohen2019certified}
bound how much a single classifier's output can change under
norm-bounded perturbations. CLEVER~\cite{weng2018clever} estimates a local
robustness score from gradient norms. $\Gamma_M$ is likewise a cheap
first-order score, but it measures a trade-off between two models (a
downstream task and a separate detector), and it is restricted to
perturbations that a compromised source can actually realize.

\textbf{Benchmarking time-series detectors.}
Wu and Keogh~\cite{wu2023flawed} show that common anomaly benchmarks
contain trivial or mislabeled anomalies. Kim et al.~\cite{kim2022rigorous}
show that point-adjusted F1 can overstate detector quality. GDN~\cite{deng2021gdn}
and TranAD~\cite{tuli2022tranad}, which we audit, are evaluated by clean
detection metrics. Our results add a separate axis: clean false-positive
rate and recall on random corruption reveal little about task-targeted
manipulation (Section~\ref{sec:blind}).

\textbf{Threat prioritization.}
Defenders rank vulnerabilities in several ways:
\begin{itemize}
  \item by severity (CVSS~\cite{first2023cvss});
  \item by likelihood of exploitation (EPSS~\cite{jacobs2021epss});
  \item by stakeholder-specific decision trees (SSVC~\cite{spring2021ssvc}),
  which recent work automates with LLMs~\cite{alhaddad2025prompting}.
\end{itemize}
Attack graphs~\cite{sheyner2002attackgraphs} rank paths through a network.
These methods prioritize \emph{known} vulnerabilities from their
descriptions or from network reachability. $\Gamma_M$ instead prioritizes
\emph{data sources} by how weakly a detector constrains task-relevant
manipulation of them, and it can be computed before any vulnerability is
reported.

\textbf{Estimation and optimal design.}
Under Gaussian residual noise with covariance $\Omega$, $A_M$ is the Fisher
information of the attack controls. $\Gamma_M^2$ in Eq.~\eqref{eq:closed} is therefore
the Fisher information on the scalar task shift
$u_M^\top b$~\cite{kay1993estimation}. Maximizing $\Gamma_M$ over added
checks is the local analogue of $c$-optimal experimental
design~\cite{pukelsheim2006optimal}. We use this classical structure to
screen and explain. The monotonicity property suggests the margin could
also guide where to add checks.

\section{Limitations and Conclusion}
Because $\Gamma_M$ is a local first-order screen rather than a global predictor of finite-budget damage across heterogeneous operating states, the pooled state-and-attack margin--harm rank correlation on Hydraulic is weak ($-.075$). Accordingly, $\Gamma_M$ is used to prioritize verification within an audit, not to estimate calibrated failure probabilities. Our evaluation covers a laboratory hydraulic rig and a five-source crop-yield panel across PCA, reconstruction, GDN, and TranAD detectors; broader architectures, physical process constraints, and time-varying source compromise remain to be evaluated.

Clean anomaly detection and resistance to targeted source manipulation are distinct properties. The margin makes the local task--detector trade-off explicit and directs finite verification toward the most weakly monitored cases.

\section*{Ethics}
All experiments use public datasets and controlled offline manipulations, and the artifact supports defensive auditing.

\section*{LLM Usage Statement}
LLMs assisted with artifact preparation and manuscript language editing.

\section*{Ethics}
All experiments use public datasets and controlled offline manipulations, and the
artifact supports defensive auditing.

\section*{LLM Usage Statement}
LLMs assisted with artifact preparation, manuscript language editing.

\bibliographystyle{IEEEtran}
\bibliography{references}

\appendices
\section{Proofs}\label{app:proofs}

\subsection{Theorem~\ref{thm:manifold}: Harmful Blind Manifold}
\begin{proof}
Let
\[
\mathcal{L}_M
  = \{b\in\R^m : G_M(b)=G_M(0)\}.
\]
Since $G_M$ has constant rank $r_M$ near zero, the constant-rank
level-set theorem makes $\mathcal{L}_M$ a local manifold of dimension
$m-r_M$, with tangent space
\[
T_0\mathcal{L}_M=\kernel J_M.
\]

Every tangent vector is the initial velocity of a local curve in
$\mathcal{L}_M$.  When $h_M>0$, choose
\[
v=\frac{\proj_{\kernel J_M}u_M}{h_M},
\qquad \|v\|_2=1.
\]
Along a curve with initial velocity $v$, the detector output remains
constant, while the chain rule gives the task's first-order change:
\[
\left.\frac{d}{dt}F_M(b(t))\right|_{t=0}
  =u_M^\top v=h_M.
\]
Reparameterizing by arc length gives the claimed unit-speed curve.
\end{proof}

\subsection{Theorem~\ref{thm:gamma}: Closed Form and Minimum-Evidence Attack}
\begin{proof}
Because $\Omega\succ0$,
\[
b^\top A_Mb=\|J_Mb\|_{\Omega^{-1}}^2,
\qquad \kernel A_M=\kernel J_M.
\]

\emph{Blind case.}
Suppose $z=\proj_{\kernel A_M}u_M\ne0$.
Then $u_M^\top z=\|z\|_2^2$, so
\[
b_0=\frac{z}{\|z\|_2^2}
\]
has $u_M^\top b_0=1$ and $J_Mb_0=0$.
Hence $\Gamma_M=0$ when $h_M>0$.

\emph{Positive-margin case.}
If $h_M=0$, then $u_M\in\range(A_M)$.  For every feasible $b$,
\begin{align*}
1=u_M^\top b
 &= (A_M^{\dagger/2}u_M)^\top(A_M^{1/2}b)\\
 &\le
 \bigl(u_M^\top A_M^\dagger u_M\bigr)^{1/2}
 \bigl(b^\top A_Mb\bigr)^{1/2},
\end{align*}
where the inequality is Cauchy--Schwarz.  Consequently,
\[
b^\top A_Mb\ge
\bigl(u_M^\top A_M^\dagger u_M\bigr)^{-1}.
\]
The vector $b^*(1)$ in Eq.~\eqref{eq:closed} is feasible and attains
this bound.  Scaling it by $\Delta$ gives the stated minimum-evidence
attack for first-order task shift $\Delta$.
\end{proof}

\subsection{Theorem~\ref{thm:frontier}: Local Frontier}
\begin{proof}
Write $w^*=b^*(1)$, so $u_M^\top w^*=1$.

\emph{Upper bound.}
The inverse-function theorem applied to
$F_M(tw^*)-F_M(0)$ gives, for small $\Delta>0$, a scalar
\[
t(\Delta)=\Delta+O(\Delta^2)
\]
such that
\[
F_M(t(\Delta)w^*)-F_M(0)=\Delta.
\]
Since $C>\|w^*\|_2$, this attack also satisfies
$\|t(\Delta)w^*\|_2\le C\Delta$ for sufficiently small $\Delta$.
Taylor expansion then yields
\begin{align*}
V_M(\Delta)
 &\le \|G_M(t(\Delta)w^*)-G_M(0)\|_{\Omega^{-1}}\\
 &=\Delta\Gamma_M+O(\Delta^2).
\end{align*}

\emph{Lower bound.}
Every feasible attack obeys $\|b\|_2\le C\Delta$.
Uniform $C^2$ Taylor bounds therefore give constants $K_F,K_G$
for which
\begin{align*}
|u_M^\top b-\Delta|&\le K_F\Delta^2,\\
\|G_M(b)-G_M(0)-J_Mb\|_{\Omega^{-1}}
  &\le K_G\Delta^2.
\end{align*}
By the definition of $\Gamma_M$ and the reverse triangle
inequality,
\begin{align*}
\|G_M(b)-G_M(0)\|_{\Omega^{-1}}
 &\ge \|J_Mb\|_{\Omega^{-1}}-K_G\Delta^2\\
 &\ge \Gamma_M|u_M^\top b|-K_G\Delta^2\\
 &\ge \Delta\Gamma_M
       -(K_G+\Gamma_MK_F)\Delta^2.
\end{align*}
Taking the infimum over feasible $b$ proves
$V_M(\Delta)=\Delta\Gamma_M+O(\Delta^2)$.
\end{proof}

\subsection{Invariance and Monotonicity}

\emph{Invariance under a change of controls.}
Let $b=\psi(c)$, where $R=D\psi(0)$ is invertible.
The Jacobian and task gradient in the new coordinates are
\[
\widetilde{J}_M=J_MR,
\qquad \widetilde{u}_M=R^\top u_M.
\]
The substitution $b=Rc$ is a bijection between the linearized
unit-effect feasible sets, and it preserves the detector evidence.
Therefore the infimum in Eq.~\eqref{eq:gamma} is unchanged.

\emph{Monotonicity under an added check.}
Suppose an independent check contributes $B_S\succeq0$ to the
information matrix.  For every $b$ satisfying $u_M^\top b=1$,
\[
b^\top(A_M+B_S)b\ge b^\top A_Mb.
\]
Taking infima over the same feasible set shows that adding the check
cannot decrease $\Gamma_M$.

\section{Realizable Telemetry Controls}\label{app:controls}

For channel $j$, let $a_j$ be the fractional gain, $\beta_j$ the
normalized bias, and $s_j$ the training standard deviation of the
channel's cycle mean.  The manipulated trace is
\[
y'_{jt}=(1+a_j)y_{jt}+s_j\beta_j.
\]
When $1+a_j\ge0$, the corresponding summaries transform as
\begin{align*}
\mu'_j
  &= (1+a_j)\mu_j+s_j\beta_j,\\
\sigma'_j
  &= (1+a_j)\sigma_j,\\
(r'_j)^2
  &= (1+a_j)^2r_j^2+s_j^2\beta_j^2\\
  &\quad +2(1+a_j)s_j\beta_j\mu_j.
\end{align*}
Here $r_j$ denotes the root-mean-square (RMS) summary.
The minimum and maximum obey the same affine transformation as the
mean.

The certificates use the derivative of this feature map at zero.
Finite attacks and alarm evaluations instead use the exact map.

\section{Numerical Robustness}

\paragraph{Certificate and derivative checks}
Certificates are cross-checked against a constrained optimizer.
MLP Jacobians are independently checked using central differences.

\paragraph{Numerical sensitivity}
On the linear crop audit, rank tolerances from $10^{-12}$ to
$10^{-6}$ and diagonal shrinkage up to $10^{-2}$ produce the same
six blind subsets.  In the audit-allocation study, identity,
diagonal, and Ledoit--Wolf full-covariance whitening give the same
median numbers of verified failures:
\[
\begin{array}{c|rrrr}
\text{Audit budget} &10&25&50&100\\ \hline
\text{Median failures} &4&9&13&15
\end{array}
\]

\paragraph{Attack feasibility}
Every returned matched attack satisfies its support and norm
constraints.  The maximum constraint residual is
$2.22\times10^{-16}$.

\end{document}